\documentclass[runningheads,a4paper]{llncs}

\usepackage{amsmath, amssymb} 
\usepackage{verbatim} 
\usepackage[pdftex]{graphicx}
\usepackage{setspace}
\usepackage{algorithmic}
\usepackage{algorithm} 
\usepackage{subfig}
\usepackage{times}
\usepackage{paralist}
\usepackage{cases}

\usepackage{paralist}
\usepackage{cases}

\usepackage{tabularx}

\usepackage{tikz}

\usepackage{color}

\newcommand{\ceiling}[1]{\left\lceil{#1}\right\rceil}

\newcommand{\kuframework}[1]{$\mathbf{k^2U}$}

 \def\myendproof{{\ \vbox{\hrule\hbox{%
   \vrule height1.3ex\hskip0.8ex\vrule}\hrule }}\par}
     
\title{Federated Scheduling Admits No Constant Speedup Factors for Constrained-Deadline DAG Task Systems\thanks{This paper has been supported by DFG, as
    part of the Collaborative Research Center SFB876
    (http://sfb876.tu-dortmund.de/).}}

\author{
    Jian-Jia Chen}
\institute{
    Department of Informatics,
    TU Dortmund University, Germany\\
}
\begin{document}

\maketitle

\vspace{-0.2in}
\begin{abstract}
  In the federated scheduling approaches in multiprocessor systems, a
  task either 1) is restricted to execute sequentially on a single
  processor or 2) has
  exclusive access to the assigned processors. There have been several
  positive results to conduct good federated scheduling policies,
  which have constant speedup factors with respect to any optimal
  federated scheduling algorithm.  This paper answers an open question: ``For
  constrained-deadline task systems with directed acyclic graph (DAG)
  dependency structures, do federated scheduling policies have a
  constant speedup factor with respect to any optimal scheduling
  algorithm?''  The answer is ``No!''  This paper presents an
  example, which demonstrates that any federated scheduling algorithm
  has a speedup factor of at least $\Omega(\min\{M, N\})$ with respect to
  any optimal scheduling algorithm, where $N$ is the number of tasks
  and $M$ is the number of processors.
  \end{abstract}

\section{Introduction}

The sporadic task model has been widely adopted in real-time systems.
In the sporadic task model, a task $\tau_i$ is characterized by its
relative deadline $D_i$, its minimum inter-arrival time $T_i$. A
sporadic task is an infinite sequence of task instances, referred to
as \emph{jobs}, where two consecutive jobs of a task should arrive no
shorter than the minimum inter-arrival time separation.
A sporadic task system $\tau$ is called an implicit-deadline
system if $D_i = T_i$ holds for each $\tau_i$. A sporadic task system
$\tau$ is called a constrained-deadline system if $D_i \leq T_i$
holds for each $\tau_i$.  Otherwise, such a sporadic task system
$\tau$ is an arbitrary-deadline system.

Traditionally, each task $\tau_i$ is also associated with its
worst-case execution time (WCET) $C_i$. In uniprocessor platforms, in
the literature, since the processor only executes one job at one
time point, there is no need to express potential parallel
execution paths.
Multiprocessor platforms allow \emph{inter-task parallelism}, which enables the
capability to execute sequential programs concurrently, and
\emph{intra-task parallelism}, which allows a parallelized task to be
executed in parallel at the same time. The parallelism of a job of a
task can be represented by a \emph{directed acyclic graph (DAG)}.
That is, the DAG structure defines the precedence constraints
of the subtasks. In other words, the execution of task $\tau_i$ 
can be divided into subtasks and the precedence constraints of these subtasks are defined by a DAG structure.
 
To handle a set of DAG tasks on multiprocessor platforms, the recent
studies by Li et al. \cite{DBLP:conf/ecrts/LiCALGS14} and Baruah
\cite{DBLP:conf/date/Baruah15,DBLP:conf/ipps/Baruah15,DBLP:conf/emsoft/Baruah15}
suggest to use federated scheduling.  In the federated scheduling in
multiprocessor systems, a task 1) either is restricted 
to execute sequentially on a single processor or 2) has exclusive access to the assigned processors.
The federated scheduling strategy was originally proposed in
\cite{DBLP:conf/ecrts/LiCALGS14} for implicit-deadline task
systems. Baruah
\cite{DBLP:conf/date/Baruah15,DBLP:conf/ipps/Baruah15,DBLP:conf/emsoft/Baruah15}
adopted the concept of federated scheduling for constrained-deadline
and arbitrary-deadline task systems. 

A scheduling algorithm generates a schedule for the task system $\tau$ to define when and how the
jobs are executed on the platform. A schedule is \emph{feasible}
if no job misses its  deadline and all the scheduling constraints are
respected (i.e., precedence constraints, minimum inter-arrival time between two
consecutive jobs of a task, etc.). An optimal scheduling algorithm is
defined as follows: \emph{If there exist a feasible schedule, an
  optimal scheduling algorithm produces one of them.} 
Similarly, an optimal \underline{federated} scheduling algorithm is
defined as follows: \emph{If there exist feasible \underline{federated} schedules, an
  optimal \underline{federated} scheduling algorithm produces one of them.}

The federated scheduling strategies are not optimal scheduling
strategies. That is, there exist task sets which can be feasibly
scheduled to meet their deadlines, but the federated scheduling
strategies lead to deadline misses while scheduling those task sets.
One widely-adopted theoretical measure to quantify the approximation
made in such \emph{non-optimal} scheduling strategies is the
\emph{speedup factors}, defined as follows:
\begin{definition}
  A scheduling algorithm ${\cal A}$ is said to have a speedup factor $s$ with respect to a scheduling algorithm ${\cal B}$ if the following condition always holds:
  \begin{itemize}
  \item For any task system $\tau$ that can be feasibly scheduled by
    the scheduling algorithm ${\cal B}$, the schedule derived from the
    scheduling algorithm ${\cal A}$ is feasible by speeding up (each
    of the processors) to $s$ times as fast as in the original
    platform (speed). \myendproof
  \end{itemize}
\end{definition}

The quantitive measure of speedup factors is always related to the
reference scheduling algorithm ${\cal B}$. The speedup factor with respect to
any optimal scheduling algorithm provides an \emph{absolute} measure to
evaluate the theoretical gap of the scheduling algorithm ${\cal
  A}$. However, if ${\cal B}$ is only a sub-optimal scheduling
algorithm, the speedup factor with respect to ${\cal B}$ provides only
a \emph{relative} measure. Therefore, if the reference algorithm ${\cal B}$
is very far from any optimal scheduling algorithm, the quantitive
speedup factors with respect to ${\cal B}$ may be misleading. 

The existing results of speedup factors for federated scheduling on $M$
identical processors can be summarized as follows:
\begin{itemize}
\item The speedup factor of the federated scheduling algorithm in
  \cite{DBLP:conf/ecrts/LiCALGS14} for implicit-deadline task systems
  in identical multiprocessor platforms is $2$ with respect to
  \emph{any optimal scheduling algorithm}.\footnote{The paper
    \cite{DBLP:conf/ecrts/LiCALGS14} uses another quantification
    metric, called capacity augmentation factor. It is also shown
    that a capacity augmentation factor $2$ also implies a speedup
    factor $2$ for implicit-deadline task systems.}
\item The speedup factor of the federated scheduling algorithms in
  \cite{DBLP:conf/date/Baruah15,DBLP:conf/emsoft/Baruah15} for
  constrained-deadline task systems in identical multiprocessor
  platforms is $3-1/M$ with respect to \emph{any optimal federated
    scheduling algorithm}.
\item The speedup factor of the federated scheduling algorithms in
  \cite{DBLP:conf/ipps/Baruah15,DBLP:conf/emsoft/Baruah15} for
  arbitrary-deadline task systems in identical multiprocessor
  platforms is $4-2/M$ with respect to \emph{any optimal federated
    scheduling algorithm}.
\end{itemize}


Therefore, there is a potential gap between the \emph{relative}
speedup factors used in
\cite{DBLP:conf/date/Baruah15,DBLP:conf/ipps/Baruah15,DBLP:conf/emsoft/Baruah15}
(with respect to any optimal federated scheduling algorithm) and the
\emph{absolute} speedup factors with respect to any optimal scheduling
algorithm. The results in
\cite{DBLP:conf/date/Baruah15,DBLP:conf/ipps/Baruah15,DBLP:conf/emsoft/Baruah15}
can only be concluded to have a constant speedup factor with respect
to any optimal scheduling algorithm if the federated schedules have a
constant speedup factor with respect to any optimal scheduling
algorithm. It could be possible that federated scheduling itself is
not a good scheduling strategy (with respect to any optimal scheduling
algorithm). If so, the constant federated speedup factors in
\cite{DBLP:conf/date/Baruah15,DBLP:conf/ipps/Baruah15,DBLP:conf/emsoft/Baruah15}
can be misleading and do not result in constant speedup factors with
respect to any optimal scheduling algorithm.

For constrained-deadline task systems with DAG, the contribution of this paper in
Section~\ref{sec:lowerbound} shows that ``the speedup factor of any
federated scheduling algorithm with respect to any optimal scheduling
algorithm is at least $\Omega(\min\{M, N\})$, where $N$ is the number
of tasks and $M$ is the number of processors.'' This concludes that
the speedup factors (with respect to any optimal scheduling
algorithm) of the algorithms in
\cite{DBLP:conf/date/Baruah15,DBLP:conf/ipps/Baruah15,DBLP:conf/emsoft/Baruah15}
are at least $\Omega(\min\{M, N\})$. However, please note that the
result in this paper does not invalidate the constant speedup factors with
respect to any optimal federated scheduling, as claimed in
\cite{DBLP:conf/date/Baruah15,DBLP:conf/ipps/Baruah15,DBLP:conf/emsoft/Baruah15}.

\section{Speedup Factor Lower Bound of Federated Scheduling}
\label{sec:lowerbound}

To prove the lower bound of the speedup factors of any federated
scheduling algorithm with respect to any optimal scheduling
algorithm, we just have to show that there exist input task sets that
admit feasible schedules but cannot be feasibly scheduled by any
federated scheduling strategies under a constant speedup
factor. Specifically, in the provided input task set, it is not
necessary to exploit any specific DAG constraints. The lower bound is
built based on the observation of the pessimistic strategy in
federated scheduling to \emph{grant a task exclusive access to the
  processors upon which they execute if the task needs more than one
  processor}.

Suppose that $M \geq 2 $ is a positive integer. Moreover, let $K$ be any arbitrary number with $K \geq 2$. 
We create $N$ sporadic tasks with the following setting:
\begin{itemize}
\item $C_1=M$, $D_1=1$, and $T_1=\infty$.
\item $C_i=K^{i-2} (K-1)M$, $D_i=K^{i-1}$, and $T_i=\infty$ for $i=2,3,\ldots, N$.
\end{itemize}
Table~\ref{tab:example} provides a concrete example when $N$ is $10$, $M$ is $10$ and $K$ is $2$.

\begin{table}[t]
  \centering
  \begin{tabular}{|c||c|c|c|c|c|c|c|c|c|c|c|}
  \hline
   & $\tau_1$   & $\tau_2$  & $\tau_3$   & $\tau_4$   & $\tau_5$   & $\tau_6$   & $\tau_7$   & $\tau_9$   & $\tau_{10}$\\
   \hline
 $C_i$ & 10 & 10 & 20 & 40 & 80 & 160 & 320 & 640 & 1280\\
\hline
 $D_i$ & 1  & 2   & 4   & 8    & 16 & 32  & 64 & 128 & 256\\
  \hline    
 $T_i$ & \multicolumn{9}{c|}{$\infty$}\\
  \hline    
  \end{tabular}
  \caption{An example of the task set $\tau$ when $N=10$, $M=10$, and $K=2$.}
  \label{tab:example}
\end{table}

We assume that each task $\tau_i$ has $M$ subtasks and there is no
precedence constraint among these $M$ subtasks (a special case of
DAG). Each subtask of task $\tau_i$ has the worst-case execution time
$\frac{C_i}{M}$.  For the rest of this section, we denote this task
set as $\tau_{counter}$.

We will first show in Lemma~\ref{lemma:feasible} that task set $\tau_{counter}$
admits feasible schedules.

\begin{lemma}
  \label{lemma:feasible} There exists a feasible schedule of the given
  task set $\tau_{counter}$.
\end{lemma}
\begin{proof}
  Since each task $\tau_i$ has $M$ (independent) subtasks with the
  same execution time, we can greedily assign each of them to one
  processor \emph{statically} and apply the earliest-deadline-first
  (EDF) scheduling algorithm individually on each of the $M$
  processors.  Therefore, the feasibility of the schedule can be
  easily verified by validating whether the subtasks on one processor
  can meet the deadline or not.  This can be verified by using the
  demand bound function analysis provided by Baruah et
  al. \cite{DBLP:conf/rtss/BaruahMR90}. Since $\frac{\sum_{i=1}^{j}
    C_i}{M} = D_j$ for $j=1,2,\ldots,N$, the above schedule is a
  feasible one. \myendproof
\end{proof}

The following lemma shows that task set $\tau_{counter}$ cannot be feasibly
scheduled by any federated scheduling algorithm if the speedup factor
$s$ is not big enough. 

\begin{lemma}
  \label{lemma:factor} 
  Suppose that the $M$ processors are speeded up to $s$ times of the
  original speed, where $s$ is strictly smaller than
  $(1-\frac{1}{K})M$, i.e., $s < (1-\frac{1}{K})M$. A federated
  schedule for task set $\tau_{counter}$ requires at least $\frac{M}{s}
  \left(N-\frac{N-1}{K}\right)$ processors with speed $s$ to feasibly
  schedule the given task set $\tau_{counter}$. That is, if $s < (1-\frac{1}{K})M$ and $\frac{M}{s}
  \left(N-\frac{N-1}{K}\right) > M$, then there is no feasible federated schedule for task set $\tau_{counter}$ on $M$ processors at such a speed $s$.
\end{lemma}
\begin{proof}
  If $\frac{C_i}{D_i s} > 1$, the concept of federated scheduling,
  i.e., \emph{tasks that are permitted to execute upon more than one
    processor are granted exclusive access to the processors upon
    which they execute}, would need to execute task $\tau_i$
  exclusively on at least $\ceiling{\frac{C_i}{D_i s}}$ processors  at speed $s$ 
  \emph{exclusively} to serve task $\tau_i$.

  For task $\tau_1$, at least $\ceiling{\frac{C_1}{D_1 s}} >
  \frac{1}{1-\frac{1}{K}} > 1$ processors are needed to ensure the
  feasibility of task $\tau_1$. Moreover, for $i=2,3,\ldots,N$, we
  have $\ceiling{\frac{C_i}{D_i s}} \geq \frac{K^{i-2} (K-1)M}{K^{i-1}
    s} = (1-\frac{1}{K})\frac{M}{s} > 1$.  Therefore, the assumption
  $s < (1-\frac{1}{K})M$ implies that a federated scheduling algorithm
  has to run
  these $N$ tasks exclusively on the granted processors.  So, task
  $\tau_i$ is assigned to be executed on at least $\ceiling{\frac{C_i}{D_i s}}$
  dedicated processors.  Therefore, if $s < (1-\frac{1}{K})M$, the number of processors in
  federated scheduling requires at least
\begin{align*}
\sum_{i=1}^{N} \ceiling{\frac{C_i}{D_i s}} &\geq \frac{M}{s} + \sum_{i=2}^{N} \frac{K^{i-2} (K-1)M}{K^{i-1} s} =\frac{M}{s}+\frac{M}{s} \sum_{i=2}^{N} \left(1-\frac{1}{K}\right)\\
&= \frac{M}{s} \left(N-\frac{N-1}{K}\right).  
\end{align*}  
  Therefore, if $s < (1-\frac{1}{K})M$ and $\frac{M}{s}
  \left(N-\frac{N-1}{K}\right) > M$, then there is no feasible federated schedule for task set $\tau_{counter}$  on $M$ processors at such a speed $s$.
\myendproof
\end{proof}

For the example task set in Table~\ref{tab:example}, suppose that the
speedup factor is $5-\epsilon$ with $\epsilon > 0$. Then, we can
conclude that $\tau_1$ needs at least three processors and each task $\tau_i$ for
$i=2,3,\ldots,10$ needs exclusively at least two processors. Therefore, at least $21$
processors are needed in this example task set under any federated
scheduling with a speedup factor $5-\epsilon$. The lower bound of
Lemma~\ref{lemma:factor} concludes that at least
$\frac{10}{5-\epsilon}(10-\frac{9}{2}) > 11$ processors are needed.
Therefore, the speedup factor of any federated scheduling algorithm with respect to any optimal scheduling algorithm by
considering this concrete example is at least $5$.  We now
conclude the lower bound of the speedup factors of any federated
scheduling algorithms with respect to any optimal scheduling algorithm.

\begin{theorem}
  \label{thm:speedup}
  The speedup factor of any federated scheduling algorithm with
  respect to any optimal scheduling algorithm for
  constrained-deadline task systems with DAG structures is at least
  $\min\left\{\left(1-\frac{1}{K}\right)M,
    \left(N-\frac{N-1}{K}\right)\right\}$.
\end{theorem}
\begin{proof}
  By Lemma~\ref{lemma:feasible}, task set $\tau_{counter}$
  admits feasible schedules.  By Lemma~\ref{lemma:factor},
 if $s < (1-\frac{1}{K})M$ and $\frac{M}{s}
  \left(N-\frac{N-1}{K}\right) > M$, then there is no feasible federated schedule for task set $\tau_{counter}$  on $M$ processors at such a speed $s$.
This implies that the resulting federated schedule under a speedup
factor $s$ with $s <
  (1-\frac{1}{K})M$ and $s < N-\frac{N-1}{K}$ is not feasible on $M$ processors. Therefore, for the task set $\tau_{counter}$,
  the speedup factor of any federated scheduling must be at least
  $\min\left\{\left(1-\frac{1}{K}\right)M,
    \left(N-\frac{N-1}{K}\right)\right\}$. \myendproof
\end{proof}

When $K$ is $2$, the speedup factor lower bound in
Theorem~\ref{thm:speedup} is at least $\min\{M/2, (N+1)/2\}$.  As a
conclusion, for the task set $\tau_{counter}$, any federated schedule
has a speedup factor at least $\Omega\left(\min\{M, N\}\right)$ with
respect to any optimal scheduling algorithm.

\section{Conclusion and Discussions}

 The result in this paper shows that \emph{at least in
  terms of the speedup metric with respect to any optimal scheduling
  algorithm, federated scheduling strategies do not yield any
  constant speedup factors for constrained-deadline task systems with
  DAG structures.}
This also invalidates the conclusions of the algorithms in
\cite{DBLP:conf/date/Baruah15,DBLP:conf/ipps/Baruah15,DBLP:conf/emsoft/Baruah15}:
\begin{quote}
  Baruah \cite{DBLP:conf/date/Baruah15,DBLP:conf/ipps/Baruah15,DBLP:conf/emsoft/Baruah15}: Our worst-case bounds indicate that at least in terms of the speedup
  metric, there is no loss in going from the three-parameter sporadic
  tasks model to the more general sporadic DAG tasks model.
\end{quote}
That is, the above conclusions in
\cite{DBLP:conf/date/Baruah15,DBLP:conf/ipps/Baruah15,DBLP:conf/emsoft/Baruah15}
stated that the DAG structures (more precisely with the option of
parallel executions) in addition to the traditional sporadic task
model (by using only three parameters $T_i, C_i, D_i$ for task
$\tau_i$ with an assumption $C_i \leq D_i$) do not introduce
additional penalty with respect to the speedup factors.  The statement
is only correct when the reference scheduling algorithm is the optimal
federated scheduling algorithms.  For the traditional sporadic task
model without parallelism, there are scheduling algorithms with a
constant speedup factor $3-1/M$ with respect to any optimal scheduling
algorithm \cite{Chakraborty2011a}.  With the example provided in this
paper, the above statement in
\cite{DBLP:conf/date/Baruah15,DBLP:conf/ipps/Baruah15,DBLP:conf/emsoft/Baruah15}
does not hold when we consider the speedup factors with respect to any
optimal scheduling algorithm.

\bibliography{biblio}{}

\end{document}